\documentclass[preprint,12pt]{elsarticle}

\usepackage{amssymb}
 \usepackage{amsthm}

\begin{document}

\begin{frontmatter}



\title{Randomized Strategy for Walking in Streets for a Simple Robot}


%
\author[mymainaddress]{Azadeh Tabatabaei\corref{mycorrespondingauthor}}
\cortext[mycorrespondingauthor]{Corresponding author}
\ead{atabatabaei@ce.sharif.edu}

\author[mymainaddress,mysecondaryaddress]{Mohammad Ghodsi }
\ead[url]{http://sharif.ir/~ghodsi/}
\address[mymainaddress]{Department of Computer Engineering, Sharif University of Technology, Tehran, Iran}
\address[mysecondaryaddress]{Institute for Research in Fundamental Sciences (IPM), Tehran, Iran}

\begin{abstract}
We consider the problem of walking in an unknown street, for a robot that has a minimal sensing capability. The robot is equipped with a sensor that  only detects the discontinuities in depth information (gaps) and can locate the target point as enters in its visibility region. First, we propose an online deterministic search strategy that generates an optimal search path for the simple robot to reach the target $t$, starting from $s$. The competitive ratio of the strategy is 9. In contrast with previously known research, the path is designed without memorizing any portion of the scene has seen so far.  Then, we present a randomized search strategy, based on  the deterministic strategy. We prove that the expected distance traveled by the robot is at most a 5.33 times longer than the shortest path to reach the target.\end{abstract}



\newtheorem{theorem}{Theorem}
\newtheorem{lemma}[theorem]{Lemma}
\newtheorem{cor}[theorem]{Corollary}
\newtheorem{prop}[theorem]{Proposition}
\newtheorem{invar}{Invariant}
\newtheorem{obs}{Observation}
\newtheorem{conj}{Conjecture}
\newtheorem{defini}{Definition}




\end{frontmatter}


\section{Introduction}

Path planning is a basic problem to almost all scopes of computer science; such as computational geometry, online algorithms, robotics and artificial intelligence \cite{GHKL}. Especially, path planning in an unknown environment for which there is no geometric map of the scene is interesting in many real life cases. Robot sensors is the only tool for gathering  information in an unknown street. Amount of the information achieved from the environment depends on the capability of the robot. Due to the importance of using simple robot, including low cost, less sensitive to failure, robust against sensing errors and noise, many types  of path planning for simple robot have been studied \cite{Bae,kao,suri1}.

 In this paper, we consider the problem of walking a simple robot in an unknown street.
A simple polygon $P$ with two separated vertices $s$ and $t$ is called a street if the left boundary chain $L_{chain}$ and the right boundary chain $R_{chain}$ constructed on the polygon from $s$ to $t$ are mutually weakly visible. In other words each point on the left chain can see at least one point on the right chain and vice versa \cite{Klein1}, see Figure \ref{f3}.
A point robot which its sensor has a minimal capability that can only detect discontinuities in depth information (gaps) and the target point $t$, starts  searching the street.  The robot can locate the target as soon as it enters in its visibility region. Also, the robot cannot measure any angles or distances or infer its position, see Figure \ref{f3}. The goal is to reach the target $t$ using the information gathered through its sensor, starting from $s$ such that the traveled path by the robot  is as short as possible.

In order to evaluate the efficiency of a search strategy for the robot, we use the notation of the competitive analysis. The competitive analysis for a strategy that leads the robot is the ratio of the expected distance traversed by the robot over the shortest distance from $s$ to $t$, in the worst case. In previous research, Tabatabaei and Ghodsi gave a deterministic algorithm for the simple robot to reach  the target $t$ in the street, starting from $s$.  The robot using some pebbles and memorizing some portion of the street has seen, explores the street. The target $t$ is achieved such that the traversed path is at most 11 times longer than the shortest by using one pebble. Also they showed, allowing use of  many pebbles reduces the factor to 9 \cite{tab}. Furthermore, they presented a randomized strategy in~\cite{EUR}.
 
In this paper, first,  we present a  deterministic strategy  using the location of two special gaps which are updated during the walking.  The robot achieves the target,  without memorizing environment and without using pebbles.  The search path is optimal; length of the generated path is at most 9 times longer than shortest path. Then, we present a  randomized    strategy that generate a search path similar to the deterministic one, but   the worst case ratio of the expected distances traveled by the robot to the shortest path is 5.33. 
%

\textbf{Related Works:} Klein proposed the first competitive algorithm for walking in streets problem  for a robot that was equipped with a 360  degrees vision system \cite{Klein1}. Also,  Icking et al. presented an optimal  search strategy for the problem with the competitive factor of  $\sqrt{2}$   \cite{IRE}. Many online strategies for patrolling unknown environment such as street, generalized street, and star polygon are presented in~\cite{GHKL,lopal}.

The limited sensing model (gap sensor) that our robot is equipped with,  in this research,  was first introduced by Tovar et al. \cite{Tovar}.  They offered Gap Navigation Tree (GNT)  to maintain and update the gaps seen along a  navigating path.  A strategies, using GNT for exploring unknown environments, presented in \cite{new}. 

Another minimal sensing model was presented  by Suri et al. \cite{suri1}. They assumed that the simple  robot can only sense the combinatorial (non-metric)
properties of the environment. The robot  can locate the  vertices of the polygon in its visibility region, and can report if there is a polygonal edge between them. Despite of the
minimal ability, they showed that the robot  can accomplish many non-trivial tasks. Then, Disser et al.   empowered  the robot with a compass to solve the mapping problem in polygons with holes \cite{disser}.
\begin{figure}
	 \centering
\includegraphics[height=7cm]{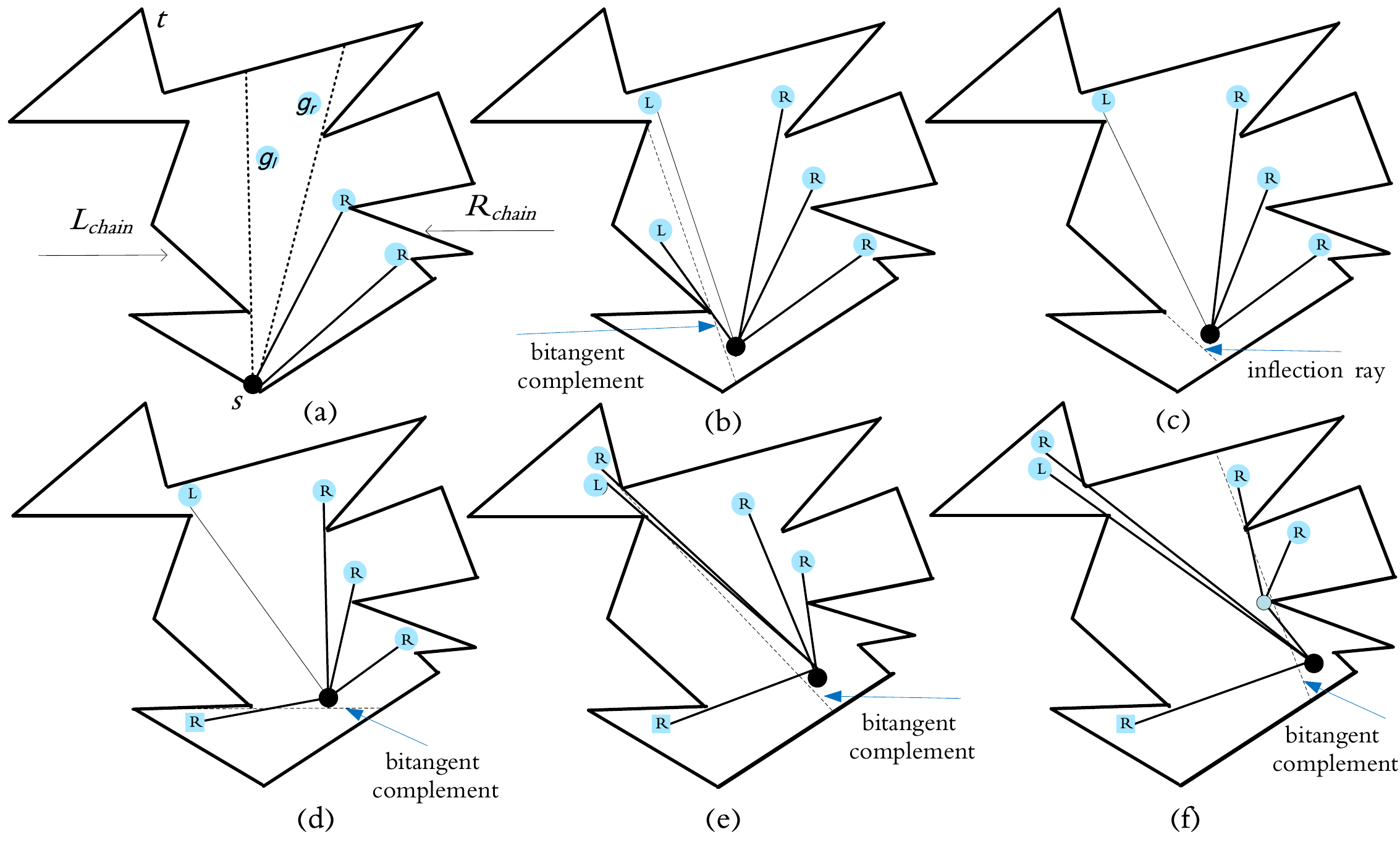}
	\caption{Street polygons, and the dynamically changes of the  gaps as the robot walks towards a gap in  street polygon. The dark circle is the location of the robot, and squares and other circles denote primitive and non-primitive gaps respectively.   (a)	Existing  gaps at the start point. (b) A split event. (c) A disappearance event. (d) An appearance event. (e) Another split event. (f) A merge event.
}
	\label{f3}
\end{figure}

\section{preliminaries}
\subsection{The Sensing Model and  motion primitives}
At the start point, the point robot reports a cyclically ordered of discontinuities in the depth information (gaps) in its visibility region. Each gap has a label of $L$ or $R$ which displays the direction of the part of the scene that is hidden behind the gap, see Figure \ref{f3}. 

 The robot can orient its heading to each gap and moves towards the gap in an arbitrary number of steps. Also the robot  moves
towards  the target as they enter in its visibility region. 

While the robot moves, combinatorial changes occur in the visibility region of the robot that they are called critical events. There are four types of critical events: appearances, disappearances, merges, and splits of gaps.  Appearance and disappearance events occur  when the robot crosses inflection rays. An  appeared  gap, during the movement,  corresponds to a portion of the environment that was already  visible, but  now is not visible. such the gaps are  called primitive gaps  and the other gaps are non-primitive gaps. Merge and split  events occur when the robot crosses bitangent, as illustrated in Figure~\ref{f3}.
\subsection{Known Properties}
At each point of the search path, if the target is not visible, the robot reports a set of left and right gaps ($l$-gap and $r$-gap for abbreviation). Let $g_l$ be the most advanced non-primitive left gap ($l$-gap) and $g_r$ be the most advanced non-primitive right gap ($r$-gap) \cite{tab}, see Figure \ref{f3}. The two gaps have a fundamental role in path planning for the simple robot.
\begin{theorem} \cite{IRE, tab}\label{the}
	While the target is not visible, it is hidden behind one of the two  gaps, $g_l$ or $g_r$.
\end{theorem}
From Theorem \ref{the}, if there exist only one of the two gaps ($g_r$ and $g_l$) then the goal is hidden behind of the gap. Thus, there is no ambiguity and the robot moves towards the gap, see Figure \ref{i22}(a). When both of $g_r$ and $g_l$ exist, a funnel case arises, see Figure \ref{i22}(b). At each funnel case, usually,    a detour from the shortest path is unavoidable.
\subsection{Essential Information}\label{me}
Whatever we  maintain  during the search strategy is location of $g_l$ and $g_r$. As the robot moves in the
street, the critical events that change
the structure of the robot's visibility region may dynamically change $g_l$ and $g_r$. Also,  by the robot movement, a  funnel case  may end or a new funnel may start. We refer to the point, in which  a funnel  ends  a \textit{critical point}  of the funnel.

The following events update the location of $g_l$ and $g_r$ as well as a funnel situation when the robot moves towards $g_l$ or $g_r$.

\begin{enumerate}
\item When     $g_r$/$g_l$ splits into $g_r$/$g_l$ and another   $r$-gap/$l$-gap, then $g_r$/$g_l$  will be replaced by the  $r$-gap/$l$-gap, (point 1 in Figure~\ref{i22}(b)).

\item When $g_r$/$g_l$ splits into $g_r$/$g_l$ and another   $l$-gap/$r$-gap, then $l$-gap/$r$-gap will be set as   $g_l$/$g_r$. This point is a critical point in which a  funnel  situation ends, (point~2 in Figure~\ref{i22}(b).

\item When $g_l$ or  $g_r$ disappears, the robot may achieve a critical point in which a   funnel situation   ends, (point~1 in Figure~\ref{i22}(b)).
\end{enumerate}
Note that the split and disappearance events may occur concurrently, (point~3 in Figure~\ref{i22}(b)). Furthermore, by moving towards $g_r$ and $g_l$, these gaps never merge with other gaps.

\section{Algorithm}
Now, we present our strategy for searching the street, from $s$ to $t$. Since the target is constantly behind one of $g_r$ and $g_l$, during the searching, the location of them is maintained and dynamically updated as  explained in the previous section.

\subsection{A deterministic strategy} At each point of the search path, especially at the start point $s$, there are two cases:
\begin{itemize}
\item
If only one of the two gaps ($g_r$ and $g_l$) existsor, or  they are collinear then the goal is hidden behind the gap. The robot moves towards the gap until the target is achieved or a funnel situation arises, see Figure \ref{i22}(a).
\item
If there is a funnel case, in order to bound the detour, the robot moves towards $g_r$ and $g_l$ alternatively.   At each   stage $i \in \lbrace 1,3,5,...\rbrace $, the robot moves  $a_i$ step(s) towards  $g_r$, and  at each   stage $i \in \lbrace 2,4,6,...\rbrace $, the robot moves $a_i$ steps  towards $g_l$ such that: $a_1=1$, $a_2=3$, and $a_i=2a_{i-1}$ for $i \geq 3$. 

 The robot continues moving towards $g_r$ or $g_l$ alternatively until a critical point of the funnel is achieved. At the point, one of $g_r$ or $g_l$ disappears, or $g_r$ and $g_l$ are collinear. So, the robot moves along the existing  gap direction until the target is achieved or a new funnel situation arises, as illustrated in Figure~\ref{i22}(b).  
\end{itemize} 
 
\subsubsection{The randomized strategy} 
Now, we present a randomized  search strategy based on the above deterministic strategy. At each point of the search path that only one of $g_r$ or $g_l$ exists, or the two gaps are collinear, the robot moves along the existing direction, similar the deterministic strategy. In the funnel case, first, the robot chooses a random real uniformly variable from $\left[  0,1\right) $ and sets  length of its step   by $2^{\varepsilon}$. Then, it chooses a uniform random variable $X$ from $\lbrace 0,1\rbrace$  to select  the direction  towards $g_r$ or $g_l$. 
If $X$ is  1/0, at each   odd stage $i$  the robot moves $a_i$ steps towards $g_r$/$g_l$, and at each  even  stage $i $, the robot moves $a_i$ steps towards $g_l$/$g_r$ ($a_n$ in the previous section). 

Similar to  the deterministic strategy, the robot continues moving towards $g_r$ and $g_l$ in the number of steps alternatively until the funnel case ends.

At each funnel, the actual randomization occurs only at first step for specifying length of steps, and for determining the direction of the movement. In the next section, we show the expected performance of our randomization algorithm is better than the  performance of our deterministic algorithm.

\begin{figure}[h]
  \centering
  \centering
\includegraphics[height=6cm]{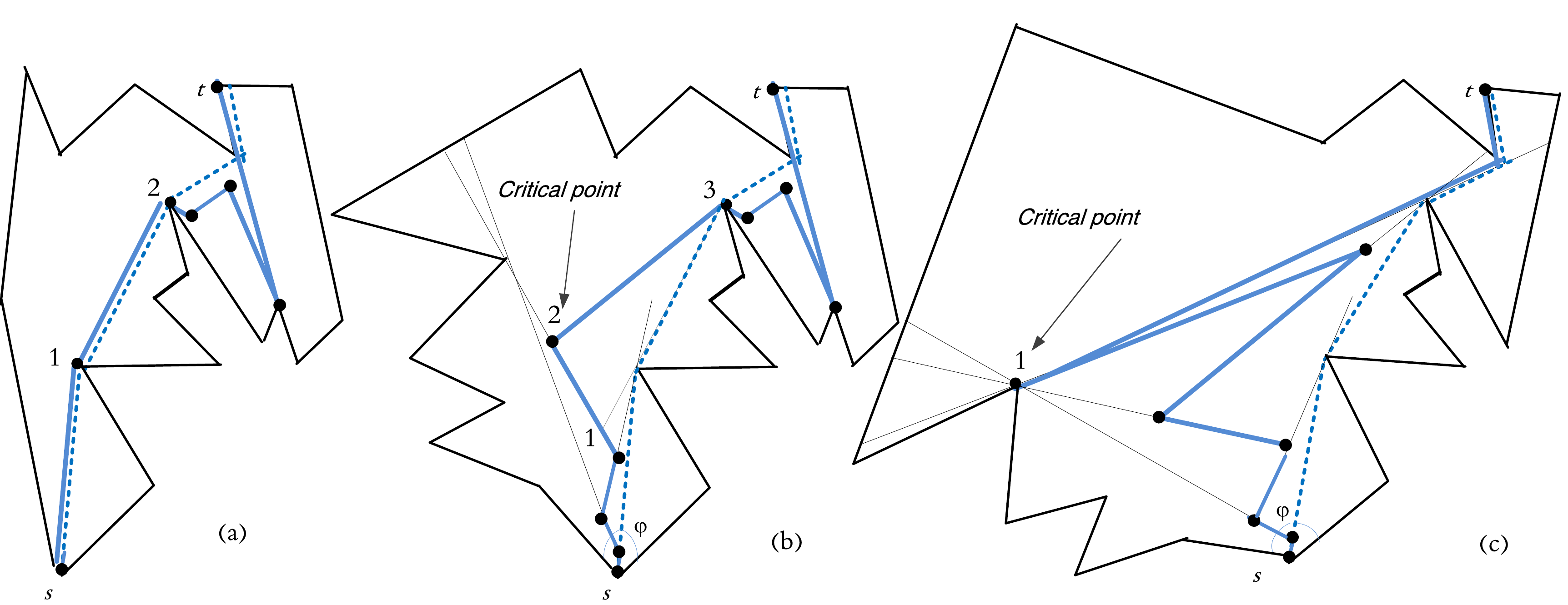}
   \caption{\small The bold path  is the robot search path, and the dotted path  is shortest path. (a) There is only $g_r$. Illustration of the algorithm for two opening angles, small and large angles respectively in (b) and (c). 
}
  \label{i22}
\end{figure}
\subsection{Correctness and Analysis}

Throughout the searching, the robot path coincides with the shortest path unless a funnel case arises. Then, in order to prove the competitive ratio of our strategy, we  compare length of the path and shortest path in a funnel case.

For analyzing the strategy, we inspire from the doubling strategy  by Baeza-Yates et al. ~\cite{Bae}.  In the strategy a robot moves back and forth on a line such that the distance to the start point doubles at each stage until the target is reached.
\begin{theorem}~\cite{Bae}~\label{XX}
The doubling strategy for searching a point on a line has
a competitive factor of 9, and this is optimal.
\end{theorem}
Opening angle, the angle between  $g_r$ and $g_l$, is always smaller than  $\pi$~\cite{IRE}. The simple robot walks towards within the opening angle. An important attribute of the angle is characterized in the following lemma.
\begin{lemma}~\label{lem}
By our strategy, the detour from shortest path for small  opening angle, in the funnel case, is  shorter than detour for large opening angle.
\end{lemma}
\begin{proof}
In each funnel case, the robot moves some steps towards $g_r$ or $g_l$, then changes its direction and moves some steps towards the  other. In the alternative movement, one of the directions is correct and the other is a deviation. Clearly for large opening angle the deviation is greater, as shown in Figure~\ref{i22}(c).
\end{proof}
Now, we can prove, the competitive factor of our deterministic strategy.
\begin{theorem}~\label{tl}
Our deterministic strategy guarantees a  path at most 9  times  longer than shortest path in the street   from $s$ to $t$.
\end{theorem}
\begin{proof}
The robot always moves towards $g_r$ and $g_l$ while the gapas  update via splitting by crossing over bitangents,  as explained in \ref{me}. Numbers of the events are finite, so the target is achieved in finite time.
From Lemma \ref{lem}, there is further deviation from shortest path for large opening angles. The angle never exceeds $\pi$. Then, for computing a competitive factor, we consider it equals $\pi$. Starting from $s$,  the robot moves $a_1=1$ step towards $g_r$, then moves   $ a_2=1+2 $ steps towards $g_l$, and again moves forth  $a_3= 2a_2=2 +2^{2} $ steps towards $g_l$,  moves back  $a_4= 2a_3=2^{2}+2^{3} $ steps towards $g_l$, and so on. In other words the robot moves back and forth on the line that contain $g_l$ and $g_r$  such that the distance to the start point  $s$ doubles at each stage until the critical point is reached. By Theorem~\ref{XX} competitive factor for the search strategy is 9.
\end{proof}

Kao, Reif, and Tate~\cite{kao} offer a  competitive randomized algorithm for searching on a line.  The procedure is similar to the doubling strategy;  the first step is characterized as a random number, and lengths of subsequent steps is multiplied  by $r$. The competitive factor of their randomized strategy is
 $1+(1+r)/\ln r$. A similar argument to the proof of theorem \ref{tl} shows  that when the opening angle is $\pi$, our  randomized search strategy coincides with Kao et al.  strategy to search a point on a line with  $r=2$. So, the  theorem  below is  immediately satisfied.

\begin{theorem}
The randomized strategy generates a search path to achieve target t in the street, starting from s, with a competitive ratio of 5.33.
\end{theorem}
\section{Conclusions}

In this paper we have developed two similar search strategies for  walking in streets problem for a simple robot. The point  robot can only detect the gaps and the target in the environment. Also the robot  can only moves towards the gaps. Our deterministic strategy achieves optimal competitive factors of 9, and is simpler than previous known result. The other strategy is a randomized strategy based on the deterministic strategy that has a better performance. Expected length of the generated path by the random strategy is at most   5.33 times longer than shortest path. It would be absorbing if there is an optimal randomized search strategy.




\end{document}